\documentclass[preprint,12pt,authoryear]{elsarticle}
\usepackage{latexsym}
\usepackage{amsmath}
\usepackage{amssymb}
\usepackage{verbatim}

\newtheorem{xdefinition}{Definition}
\newtheorem{xobservation}{Observation}
\newtheorem{xtheorem}{Theorem}
\newtheorem{xlemma}{Lemma}
\newtheorem{xproposition}{Proposition}
\newtheorem{xcorollary}{Corollary}
\newenvironment{definition}{\begin{xdefinition}\rm}%
{\hspace*{\fill}\raisebox{-1pt}{\boldmath$\Box$}\end{xdefinition}}
{\hspace*{\fill}\raisebox{-1pt}{\boldmath$\Box$}\end{xobservation}}
\newenvironment{theorem}{\begin{xtheorem}\rm}{\end{xtheorem}}
\newenvironment{lemma}{\begin{xlemma}\rm}{\end{xlemma}}
\newenvironment{proposition}{\begin{xproposition}\rm}{\end{xproposition}}
\newenvironment{corollary}{\begin{xcorollary}\rm}{\end{xcorollary}}
\newenvironment{proof}{\begin{trivlist}\item[]{\bf Proof }}%
{\hspace*{\fill}\raisebox{-1pt}{\boldmath$\Box$}\end{trivlist}}

\newcommand{\bbF}{\mathbb{F}}
\newcommand{\mapnbit}{\colon \bbF_2^n\rightarrow \bbF_2}
\DeclareMathOperator{\spa}{span}

\journal{Theoretical Computer Science}

\begin{document}
\begin{frontmatter}
 \title{Multiplicative Complexity of Vector Valued Boolean Functions}

\author[IMADA]{Joan Boyar\fnref{fn1}}
\ead{joan@imada.sdu.dk}

\author[IMADA]
{Magnus Gausdal Find}
\ead{magnus@gausdalfind.dk}

\address[IMADA]{Department of Mathematics and Computer Science,
University of Southern Denmark,
 Campusvej 55, DK-5230 Odense M, Denmark}

\fntext[fn1]{Corresponding author.}

\begin{keyword}
Multiplicative complexity \sep
Nonlinearity \sep
Circuits \sep
Error correcting codes 
\end{keyword}
\begin{abstract}
We consider the multiplicative complexity of Boolean functions with multiple bits of output, studying how large a multiplicative complexity is necessary and sufficient to provide a desired nonlinearity. 
For so-called $\Sigma\Pi\Sigma$ circuits, we show that 
there is a tight connection between error correcting codes and circuits
computing functions with high nonlinearity.
Combining this with known coding theory results, we show that functions with $n$ inputs and $n$ outputs with the highest possible nonlinearity must have at least $2.32n$ AND gates.  
We further show that one cannot prove stronger lower bounds
by only appealing to the nonlinearity of a function;
we show a bilinear circuit
computing a function with almost optimal nonlinearity with
the number of AND gates being exactly the length of such a shortest code.

Additionally we provide a function which, for general circuits, has multiplicative complexity at least $2n-3$. 

Finally we study the multiplicative complexity of ``almost all'' functions. We show that every function with $n$ bits of input and $m$ bits of output can be computed using at most $2.5(1+o(1))\sqrt{m2^n}$ AND gates.
\end{abstract}
\begin{keyword}
Multiplicative complexity \sep
Nonlinearity \sep
Circuits \sep
Error correcting codes \sep
\end{keyword}

\end{frontmatter}

\section{Introduction}
Cryptographic functions such as encryption functions should have
high nonlinearity to be resistant against linear and differential
attacks (see again \cite{carletvectorial} and the references therein). This is an explicit design criteria for modern cryptographic systems, such as
AES,~\citep{daemen2002design}, which has been used as a benchmark
for several implementations of homomorphic encryption.

In several settings, such as homomorphic encryption and secure multiparty computation (see e.g. \cite{DBLP:conf/focs/Vaikuntanathan11} and \cite{kolesnikov}), for practicality/efficiency, the number of AND gates is significantly more important than the number of XOR gates, hence one is interested in functions with as few AND gates as possible. For many such protocols, it is not just the \emph{number} of AND gates that matters, but also the \emph{AND depth}. For example, in several protocols for secure multiparty computation the number of AND gates is proportional to the number of bits sent, and the AND depth corresponds to the number of rounds in the protocol (see e.g. \cite{LP13}), and in typical protocols for homomorphic encryption, the norm of the noise
after an AND gate is the product of the norms of the noise from the inputs, so the AND depth greatly affects the number of expensive bootstrappings, relinearizations, and/or modulus reductions which are necessary (see e.g. \cite{DSES14}). For more examples we refer to \cite{ARSTZ15} and the references therein.

A natural question to ask is how the nonlinearity of a function and its multiplicative complexity (the number of AND gates necessary to compute it when only AND, NOT and XOR gates are used) are related to each other: how large does one measure need to be in order for the other to have at least a certain value? 
As stated in Section \ref{sec:defsnprelims}, for every desired nonlinearity,
it is known exactly how many AND gates are necessary and sufficient for 
functions with only one output to achieve this.
We study this same question for functions with multiple bits
 of output.

\subsection{Definitions and Preliminaries}
\label{sec:defsnprelims}
Let $\bbF_2$ be the finite field of order $2$ and $\bbF_2^n$
the $n$-dimensional vector space over $\bbF_2$.

We denote by $[n]$ the set $\{ 1, \ldots, n \}$.
An $(n,m)$-function is a mapping from
$\bbF_2^n$ to $\bbF_2^m$ and we refer to these
as the  \emph{Boolean functions}. When $m>1$ we say that the function is \emph{vector valued}.

It is well known that every $(n,1)$-function $f$ can be
written uniquely as a multilinear polynomial over $\bbF_2$
\[
f({x}_1,\ldots ,{x}_n)=
\sum_{X\subseteq [n]}\alpha_{X}\prod_{i\in X}{x}_i,
\]
where $\alpha_X\in \{ 0,1\}$ for subsets of indices.
This polynomial is called the \emph{Zhegalkin polynomial} or the
\emph{algebraic normal form} (ANF) of $f$. For the rest of this paper
most, but not all, arithmetic will be in $\bbF_2$.
We trust that the reader will find it clear
whether arithmetic is in $\bbF_2$, $\bbF_{2^n}$, or $\mathbb{R}$
when not explicitly stated, and will not address it further.

The \emph{degree} of an $(n,1)$-function $f$
is the largest $|X|$ such that $\alpha_X=1$. 
For  an $(n,m)$-function $f$, we let $f_i$ be the $(n,1)$-function
defined by the $i$th output bit of $f$, and say that the degree
of $f$ is the largest degree of $f_i$ for $i\in [m]$.
A function is \emph{affine} if it has degree $1$,
and \emph{quadratic} if it has degree $2$.
For $T\subseteq [m]$
we let
\[
f_T=\sum_{i\in T}f_i,
\]
and for  $\mathbf{v}\in \bbF_2^n$ we let $|\mathbf{v}|$ denote the 
\emph{Hamming weight} of $\mathbf{v}$, that is, the number of nonzero
entries in $\mathbf{v}$, and let $|\mathbf{u}+\mathbf{v}|$ be the
\emph{Hamming distance} between
the two vectors $\mathbf{u}$ and $\mathbf{v}$.

\paragraph{Nonlinearity of Boolean Functions}
We will use several facts on the nonlinearity of Boolean functions. We
refer to the two chapters in~\cite{carletpreds,carletvectorial} for proofs and references.

The \emph{nonlinearity} of an $(n,1)$-function $f$
is the Hamming
distance to the closest affine function, more precisely
\[
NL(f)= 
2^{n}- \max_{\mathbf{a}\in \bbF_2^n,b\in \bbF_2}|\{\mathbf{x}\in\bbF_2^n |
\left\langle \mathbf{a}, \mathbf{x}\right\rangle+b=f(\mathbf{x})\}|,
\]
where $\left\langle \mathbf{a}, \mathbf{x}\right\rangle=\sum_{i=1}^n a_i x_i$.
For an $(n,m)$-function $f$, the nonlinearity is defined as
\[
NL(f) = \min_{T\subseteq [m],T\neq \emptyset}\{NL(f_T)\}.
\]
The nonlinearity of an $(n,m)$-function is always
between $0$ and $2^{n-1}-2^{\frac{n}{2}-1}$.
The $(n,m)$-functions which meet the upper bound
are called \emph{bent functions}. Bent $(n,1)$ functions exist
if and only if
$n$ is even.
A standard example 
of a bent $(n,1)$-function is the \emph{inner product}, on $n=2k$
variables, defined as:
\[ 
IP_{2k}(
x_1,\ldots , x_{k}
,
y_1,\ldots , y_{k})
=
\left\langle \mathbf{x}, \mathbf{y}\right\rangle.
\]
This function is clearly quadratic.
If we identify $\bbF_2^n$ with $\bbF_{2^n}$, a standard example of a
bent $(2n,n)$-function is the \emph{finite field multiplication} function:
\begin{equation} 
\label{eq:finitefieldmult}
 f(\mathbf{x},\mathbf{y})=\mathbf{x} \cdot \mathbf{y}
\end{equation}
where multiplication is in $\bbF_{2^n}$. 

If $n=m$, $NL(f)$ is at most $2^{n-1}-2^{\frac{n-1}{2}}$, see \cite{CV94}.  Functions meeting this bound
are called \emph{almost bent}. These exist only for odd $n$.
As remarked by Carlet, this name is a bit misleading since the name indicates
that they are suboptimal, which they are not. Again, if
 we identify
$\bbF_2^n$ and $\bbF_{2^n}$, for
$1\leq i\leq \frac{n-1}{2}$ and $gcd(i,n)=1$, the so called \emph{Gold functions}, $G(x)$, defined
as
\begin{equation}
\label{eq:almostbentexample}
G(\mathbf{x})=\mathbf{x}^{2^{i}+1}=\mathbf{x}\cdot \left (\mathbf{x}^{2^{i}}\right)
\end{equation}
are almost bent. This function is quadratic since the $\bbF_{2^n}$-operator $\mathbf{x}\mapsto \mathbf{x}^2$ is affine
in when considered as an operator on $\bbF_2^n$, and each output bit of
finite field multiplication is quadratic in the inputs,
see also \cite{carletvectorial}.

\paragraph{Multiplicative Complexity and Circuit Classes}
In this paper we consider multiple classes of circuits:

\begin{itemize}
\item
An \emph{XOR-AND circuit} is a Boolean circuit where each of the gates is either $\oplus$ (XOR, addition in $\bbF_2$), $\wedge$
(AND, multiplication in $\bbF_2$) or the constant $1$. The XOR gates may have unbounded fanin, and the AND gates have fanin $2$;
\item a $\Sigma\Pi\Sigma$ circuit is a circuit containing three layers of gates. The first layer contains XOR gates, the second contains AND gates and the third contains XOR gates. All gates are allowed unbounded fanin;
\item
a circuit is \emph{quadratic} if it is both a $\Sigma\Pi\Sigma$ circuit and an XOR-AND circuit. This is equivalent to saying that all gates in the circuit compute functions of degree at most 2;
\item
a quadratic circuit is \emph{bilinear} if the input
is partitioned into two sets, and each input to an AND gate is a 
linear combination of variables from one of these two sets, with
the other input using the opposite set of the partition.
\end{itemize}
The \emph{multiplicative complexity} of an $(n,m)$-function with respect to a circuit class $\mathcal{C}$ is the smallest number of AND gates in any circuit from $\mathcal{C}$ computing $f$. If the circuit class $\mathcal{C}$ is not specified, we refer to the multiplicative complexity with respect to XOR-AND circuits.

We notice the hierarchy between the circuits: any bilinear circuit is a quadratic circuit and any quadratic circuit is both an $\Sigma\Pi\Sigma$ and an XOR-AND circuit.

Some functions have higher multiplicative complexity with respect to XOR-AND circuits than to $\Sigma\Pi\Sigma$ circuits. Examples of this include the function $f(\mathbf{x})=x_1x_2 \cdots x_n$, having multiplicative complexity, $n-1$ and $1$, respectively. On the other hand, some functions have much lower multiplicative complexity with respect to XOR-AND circuits than to $\Sigma\Pi\Sigma$ circuits. An example of this is the \emph{majority} function, known to have multiplicative complexity close to $n$ \citep{DBLP:journals/tcs/BoyarP08} with respect to XOR-AND circuits, but exponentially many AND gates are needed in any $\Sigma\Pi\Sigma$ circuit (see \cite{razborov1987lower} and \cite{juknabook}). Note that the exponential lower bound on the total number of gates in a $\Sigma\Pi\Sigma$ circuit implies an exponential lower bound on the number of AND gates. To see this, note that the output level contains only one gate, and if there were more than $n$ inputs to any AND gate, they would be linearly dependent. The AND of linearly dependent linear functions can either be expressed with fewer inputs or has an output of zero and can be removed.

\paragraph{Relationship Between Nonlinearity and Multiplicative Complexity}

Some relations between nonlinearity and multiplicative complexity are
known (see also \cite{carlet2001complexity}).

\begin{proposition}[\cite{DBLP:journals/tcs/ZhengZI99}]
\label{zhengetal}
Let $f$ be an $(n,1)$-function. Suppose there exists an affine subspace $U$ of dimension $k$ such that $f$ is affine on the points in $U$. Then $f$ has nonlinearity at most $2^{n-1}-2^{k-1}$.
\end{proposition}

If a function has multiplicative complexity $n-k$, then there exists an affine subspace of dimension $k$ on which $f$ is affine (see e.g. \cite{DBLP:journals/tcs/BoyarP08}). Applying this the following proposition is immediate. 

\begin{proposition}
\label{prop:nlvsmcpredicates}
Let $f$ be an $(n,1)$-function with multiplicative complexity $M$. Then the $NL(f)\leq 2^{n-1}-2^{n-M-1}$.
\end{proposition}
We remark that this result holds for both $\Sigma\Pi\Sigma$ circuits and XOR-AND circuits. It is a slight generalization of a result given in \cite{DBLP:conf/ciac/BoyarFP13} where we showed a similar relation. The proof of the latter result holds however only for XOR-AND circuits. There it is also shown that the above results are tight; for every $M\leq \frac{n}{2}$ there exists an explicit function meeting the bound with equality, for both $\Sigma\Pi\Sigma$ and
XOR-AND circuits.

\paragraph{Linear Codes}
\label{sec:lincodes}
Most bounds in this paper will come from coding theory. 
In this subsection, we briefly review the necessary facts. For more information,
see chapter 17 in \cite{jukna2011extremal}
or the older but comprehensive
\citep{sloane1977theory}.

A \emph{linear (error correcting) code} of \emph{length} $s$ is a linear subspace, $\mathcal{C}$ of
$\bbF_2^s$.
The \emph{dimension} of a code is the dimension of the subspace, $\mathcal{C}$,
 and the elements of $\mathcal{C}$ are called \emph{codewords}.
The (minimum) \emph{distance} $d$ of $\mathcal{C}$ is defined as
\[
d= \min_{\mathbf{x\neq y}\in \mathcal{C}, } |\mathbf{x}+ \mathbf{y}|.
\]
The following fact is well known
\begin{proposition}
\label{prop:linearcodesdistance}
  For every linear code, $\mathcal{C}$, the distance is exactly the minimum
weight among non-zero codewords.
\end{proposition}

Let $L(m,d)$ be the length of the shortest linear $m$-dimensional code over $\bbF_2$ with distance $d$.
We will use lower and upper bounds on $L(m,d)$.
One lower bound is the following
\citep{DBLP:journals/tit/McElieceRRW77}, see also \cite{sloane1977theory},
page 563.

\begin{theorem}[McEliece, Rodemich, Rumsey, Welch]
\label{thm:mrrw2}
For $0 < \delta < 1/2$, let $\mathcal{C}\subseteq \{0,1\}^{s}$ be a linear 
 code with dimension $m$ and
distance $\delta s$. Then the rate $R=\frac{m}{s}$ of the code
satisfies
$R\leq \min _{0\leq u\leq 1-2\delta }B(u,\delta)$, 
where
$B(u,\delta)=1+h(u^2)-h(u^2+2\delta u+2\delta)$, 
$h(x)=H_2\left( \frac{1-\sqrt{1-x}}{2}\right)$, and $H_2(x) = -x \log x - (1-x)\log (1-x)$.
\end{theorem}
An upper bound is the following, see
\cite{jukna2011extremal}.

\begin{theorem}[Gilbert-Varshamov]
\label{thm:gilbertvar}
  A linear code $\mathcal{C}\subseteq \{0,1 \}^s$ of dimension $m$
and distance $d$ exists provided that
$
\sum_{i=0}^{d-2}{s-1\choose i} < 2^{s-m}.
$
\end{theorem}

\subsection{Results}

\paragraph{Our Contributions}
Let $f$ be an $(n,m)$-function with nonlinearity $2^{n-1}-2^{n-M-1}$.
We show that any $\Sigma\Pi\Sigma$ circuit with $s$ AND gates computing $f$ defines an $m$-dimensional linear code in $\bbF_2^s$ with distance $M$, so lower bounds on the size of such codes show lower bounds on the number of AND gates in such a circuit.
This means that a proof of high nonlinearity for a function
is automatically a lower bound on the multiplicative complexity for this class of circuits. We instantiate this result for the specific case of quadratic circuits. First, the so called \emph{Gold} functions with $n$ bits of input and $n$ bits of output are quadratic and almost bent.
Using known coding theory bounds, we 
conclude that any quadratic circuit computing such functions must have at least $2.32n$ AND gates. To the best of our knowledge this is the best lower bound for such circuits.
Similarly a well known lower bound for finite field multiplication 
\citep{DBLP:journals/tc/BrownD80,DBLP:journals/tcs/LempelSW83}
follows as a direct corollary.

On the other hand, we show that appealing only to the nonlinearity of a
function cannot lead to much stronger lower bounds on the multiplicative complexity by showing the existence
of \emph{quadratic} (in fact, \emph{bilinear}) \emph{circuits} 
from $n$ bits to $n$ bits with nonlinearity
at least $2^{n-1}-2^{\frac{n}{2}+3\sqrt{n}}$ containing $s$ AND gates
where $s$ is the length of a shortest $n$-dimensional code of distance $\frac{n}{2}$.
The Gilbert-Varshamov bound gives that $s\leq 2.95n$.

Although almost all Boolean functions with $n$ inputs and one output
have multiplicative complexity at least $2^{n/2}-O(n)$ \citep{DBLP:journals/tcs/BoyarPP00}, no concrete function
of this type has been shown to have multiplicative complexity more than
$n-1$.
We give a concrete function with $n$ inputs and $n$ outputs with multiplicative 
complexity at least $2n-3$. To the best of our knowledge this is the best lower bound for the multiplicative complexity for an explicit $(n,n)$-function.

Finally we study the worst case multiplicative complexity of vector valued functions. We show that almost every $(n,m)$-function has multiplicative complexity at least $(1-o(1))\sqrt{m2^n}$ and that \emph{every} such function can be computed using an XOR-AND circuit with at most $(2.5+o(1))\sqrt{m2^n}$ AND gates.

\paragraph{Related Results}
\label{sec:relatedresults}
Previous results showing relations between error correcting codes and 
bilinear and quadratic circuits include
 the work of \citep{DBLP:journals/tc/BrownD80,DBLP:journals/tcs/LempelSW83}
where
it is shown that a bilinear
or quadratic circuit computing finite field multiplication 
of two $\bbF_{q^n}$ elements 
induces an error correcting
code over $\bbF_q$ of dimension $n$ and distance $n$.
For $q=2$, Theorem~\ref{thm:mrrw2} implies that such a circuit must have at least $3.52n$ multiplications (AND
gates).
If $n$ is the number of input bits, this corresponds to a lower bound of $1.76n$.
 Note that this result can be obtained as a corollary of our results,
which relate the nonlinearity of a quadratic function with $m$ outputs and
multiplicative complexity $s$ to the distance of an $m$-dimensional linear 
code over vectors of length $s$.

For $q>2$, the gates (or lines in a straight-line program) have field elements
as inputs, and the total number of multiplications and divisions is counted.
A lower bound of $3n -o(n)$ for bilinear circuits was shown in~\cite{DBLP:journals/jacm/KaminskiB89},
and it was extended to general circuits in \citep{DBLP:journals/cc/BshoutyK06}.
This proof is not based on coding theoretic techniques, but rather the
study of Hankel matrices related to the bilinear transformation.

A different relation between general Boolean (not just quadratic) $(n,1)$-functions and error correcting codes comes from the following observation:
Suppose some $(n,m)$-function $f$ has a certain nonlinearity $D$.
If we identify the functions $f_1,\ldots, f_m$, $x_1,\ldots,x_n$ and the constant $1$ with their
truth tables as vectors in $\bbF_2^{2^n}$, then
$\mathcal{C}=\spa\{f_1,\ldots, f_m,x_1,\ldots,x_n,1\}$ is a code in $\bbF_2^{2^n}$
with dimension $n+m+1$ and distance $D$, and limitations and possibilities
for codes
transfer to results on nonlinearity (see
the survey 
\citep{carletvectorial} and the references therein).
However this says nothing about the multiplicative complexity of the function $f$.

The structure of quadratic circuits has itself been studied 
by Mirwald and Schnorr \citep{DBLP:journals/tcs/MirwaldS92}. Among
other things they show that for quadratic $(n,1)$- and $(n,2)$-functions, 
quadratic circuits are optimal. It is still not known
whether this is true for $(n,m)$-functions in general.

\section{$\Sigma\Pi\Sigma$ Circuits: Multiplicative Complexity and Nonlinearity}
This section is devoted to showing a relation between
the nonlinearity and the multiplicative complexity of $\Sigma\Pi\Sigma$ circuits.
We first show a connection between nonlinearity, multiplicative complexity
and certain linear codes.
Applying this connection, Theorem~\ref{thm:mrrw2} gives a relationship between
nonlinearity and multiplicative complexity for any
quadratic $(n,m)$-function.

\begin{theorem}
\label{thm:nlecc}
Let the $(n,m)$-function, $f$, have
$NL(f)\geq 2^{n-1}-2^{n-M-1}$, where $M\leq \frac{n}{2}$.
Then a $\Sigma\Pi\Sigma$ circuit with $s$ AND gates computing $f$ exhibits an $m$-dimensional linear
code over $\bbF_2^s$ with distance $M$.   
\end{theorem}

\begin{proof}
  Let $C$ be a $\Sigma\Pi\Sigma$ circuit with $s$ AND gates computing $f$,
and let $A_1,\ldots, A_s$, be the AND gates.
Since $C$ is $\Sigma\Pi\Sigma$, for each $i\in [m]$ there
exist $S_i\subseteq [s]$ and $X_i\subseteq [n]$ such that $f_i$ can be 
written as
\[
f_i = \sum_{j\in S_i}A_j+\sum_{j\in X_i}{x}_{j}.
\]
Without loss of generality, we can assume that $X_i=\emptyset$ for all $i$, since
both nonlinearity and multiplicative complexity are invariant
under the addition of affine terms.
For each $i\in [m]$, we define the vector $\mathbf{v}_i\in \bbF_2^s$,
where $\mathbf{v}_{i,j}=1$ if and only if there is a directed path
from $A_j$ to the $i$th output.
By the nonlinearity of $f$, we have that for each $i\in [m]$, 
\[
NL(f_i)\geq 2^{n-1}-2^{n-M-1}.
\]
Applying Proposition~\ref{prop:nlvsmcpredicates},
the multiplicative complexity of $f_i$ is at least 
$M$, hence $|\mathbf{v}_i|\geq M$.
Similarly, for any nonempty $T\subseteq [m]$ we can associate a vector
$\mathbf{v}_T$ by setting
\[
\mathbf{v}_T=\sum_{i\in T}\mathbf{v}_i.
\]
Since the circuit is $\Sigma\Pi\Sigma$, it holds that if $|\mathbf{v}_T|\leq p$,
the multiplicative complexity of
$f_T=\sum_{i\in T}f_i$ is at most $p$.
Applying the definition of nonlinearity to $f_T$,
$NL\left(f_T\right)\geq 2^{n-1}-2^{n-M-1 }$.
Proposition~\ref{prop:nlvsmcpredicates} implies that the multiplicative complexity
of $f_T$ is at least $M$, so we have that 
$|\mathbf{v}_T|\geq M$ when $T\neq \emptyset$.

In conclusion, every nonzero vector in 
the $m$ dimensional vector space $\mathcal{C}=\spa_{\bbF_2}\{\mathbf{v}_1,\ldots,\mathbf{v}_m \}$
has Hamming weight at least $M$. By Proposition~\ref{prop:linearcodesdistance}, $\mathcal{C}$ is
a linear code with dimension $m$ and distance at least $M$.
\end{proof}

Applying this theorem to quadratic almost bent functions, we have that a quadratic circuit
computing such a function has at least $L(n,\frac{n-1}{2})$ AND gates. Combining this with Theorem~\ref{thm:mrrw2}, calculations, which we include for the
sake of completeness, show:
\begin{corollary}
\label{cor:almostbentcorollary}
Any quadratic circuit computing an almost bent $(n,n)$-function
has at least $L(n,\frac{n-1}{2})$ AND gates. For sufficiently large
$n$, $L(n,\frac{n-1}{2})>2.32n$.
\end{corollary}

\begin{proof}
Recall that almost bent functions have nonlinearity $2^{n-1}-2^{\frac{n-1}{2}}$,
 so
in terms of Theorem~\ref{thm:nlecc}, we have $M=\frac{n-1}{2}$.
Suppose for the sake of contradiction
that $L(n,\frac{n-1}{2})\leq 2.32n$ for infinitely many values of $n$.
Then we have an $n$-dimensional code on $\bbF_2^{2.32n}$ with distance
$\frac{n-1}{2}$. From Theorem~\ref{thm:mrrw2}, we have $\delta =
\frac{n-1}{2\cdot 2.32n}$, so for sufficiently large $n$,
\[
\delta > 0.2155. 
\]
The rate of the
code is  $R=\frac{1}{2.32}> 0.431$. 
Choosing $u=0.32$ shows that the rate $R$ satisfies
\[
0.431<R\leq B(0.32,0.2155)< 0.42,
\]
contradicting the assumption that $L(n,\frac{n-1}{2})\leq 2.32n$.
\end{proof}

The corollary above applies to e.g. the almost bent
Gold functions $G$ defined in Eqn.~\ref{eq:almostbentexample}. 
For bent $(2n,n)$-functions, using Theorem~\ref{thm:nlecc} with
$M=n$ and applying Theorem~\ref{thm:mrrw2}, calculations, which are
not new, but again included for the sake of completeness, show: 

\begin{corollary}
\label{cor:bent2nn}
A quadratic circuit computing any bent $(2n,n)$-function has at 
least $L(n,n)$ AND gates. For sufficiently large $n$, $L(n,n)>3.52n$.
\end{corollary}

\begin{proof}
As in the proof above, in terms of Theorem~\ref{thm:nlecc}, we have
$m=n$ and $M=n$. 
Again suppose for the sake of contradiction that
$L(n,n)\leq 3.52$ for infinitely many values of $n$.
In Theorem~\ref{thm:mrrw2}, for sufficiently large $n$ we have that
\[
R=\delta=\frac{1}{3.52}  \approx 0.28409.
\]
Choosing $u=0.4$ in 
Theorem~\ref{thm:mrrw2} shows that the rate $R$ satisfies
\[
0.28409<R\leq B(0.4, 0.28409 ) \approx 0.28260 < 0.284.
\]
contradicting the assumption that $L(n,n)\leq 3.52$.
\end{proof}

This applies to e.g. the finite field multiplication function as defined
in Eqn.~\ref{eq:finitefieldmult},
reproving the known result on multiplicative complexity for quadratic
circuits for field multiplication mentioned in Section~\ref{sec:relatedresults}.

For both Corollaries~\ref{cor:almostbentcorollary} and \ref{cor:bent2nn},
any improved lower bounds on codes lengths would give an improved lower bound
on the multiplicative complexity.
For Corollary~\ref{cor:almostbentcorollary} 
this technique cannot prove substantially better lower bounds than
$L(n,\frac{n-1}{2})$. Theorem~\ref{thm:gilbertvar} implies that
$L(n,\frac{n-1}{2})\leq 2.95n$.
Below we show that this is not merely a limitation of the proof strategy;
there exist quadratic circuits with $L(n,\frac{n-1}{2})$ AND gates with
nonlinearity relatively close to the optimal. To the best of our knowledge this is the first
example of highly nonlinear $(n,n)$-functions with linear multiplicative complexity,
and therefore it might be a useful building block for cryptographic purposes.

Before proving the next theorem, we need a technical lemma on the
probability that a random matrix has small rank. A simple
proof of this can be found in e.g. \cite{DBLP:conf/focs/KomargodskiRT13}. 

\begin{lemma}[Komargodski, Raz, Tal]
\label{avishayrankprob}
  A uniform random $k\times k$ matrix over $\bbF_2$ has rank at most $d$ with 
probability at most 
$2^{k-(k-d)^2}$.
\end{lemma}

\begin{theorem}
  There exist $(n,n)$-functions
with multiplicative complexity at most $L(n,\frac{n-1}{2})$
and nonlinearity at least $2^{n-1}-2^{\frac{n}{2}+{3}\sqrt{n}-1}$.
\end{theorem}
\begin{proof}
For simplicity we show the upper bound for $L(n,\frac{n}{2})$ AND gates. It is elementary
to verify that it holds for $L(n,\frac{n-1}{2})$ AND gates as well.
We give a probabilistic construction of a 
quadratic (in fact, bilinear) circuit with $s=L(n,\frac{n}{2})$ AND gates, then we show that with high
probability, the function computed by this circuit
has the desired nonlinearity. 

For the construction of the circuit, we first define the value
computed by the $i$th AND gate as
$A_i(\mathbf{x})=L_i(\mathbf{x})R_i(\mathbf{x})$
where $L_i$ is a random sum over ${x}_1,\ldots ,{x}_{n/2}$
and $R_i$ is a random sum over ${x}_{n/2+1},\ldots ,{x}_n$.
In the following, we will identify
sums over $\mathbf{x}_1,\ldots , \mathbf{x}_n$ with
vectors in $\bbF_2^n$ 
and
sums over $A_1,\ldots ,A_s$ with vectors in $\bbF_2^s$.

Let $\mathcal{C}$ be an $n$-dimensional code of length $L(n,\frac{n}{2})$ with 
distance $\frac{n}{2}$ and let 
$\mathbf{y}_1,\ldots,\mathbf{y}_n \in \bbF_2^s$
be a basis for $\mathcal{C}$. Now we define the corresponding
sums over $A_1,\ldots ,A_s$ to be the outputs computed by the circuit. 
This completes the construction of the circuit. 
Now fix
$r(\mathbf{x})\in \spa_{\bbF_2}\{\mathbf{y}_1,\ldots ,\mathbf{y}_n\}$,
$r\neq \mathbf{0}$. We want to show that $r$ has the desired nonlinearity
with high probability.
By an appropriate relabeling of the AND gates, we can write $r$ as
\begin{equation}
\label{eq:rexpression}
r(\mathbf{x})= \sum_{i=1}^qA_i(\mathbf{x})=
\sum_{i=1}^qL_i(\mathbf{x})R_i(\mathbf{x})
\end{equation}
for some $q\geq \frac{n}{2}$. We now assume that
\begin{equation}
  \label{eq:assumptionhighrank1}
t=\mathrm{rk}\{ R_1,\ldots , R_{q} \} \geq \frac{n}{2}-\frac{3\sqrt{n}}{2}.
\end{equation}
At the end of the proof, we will show that this is true with high probability.
Again by an appropriate relabeling, we let
$\{R_1,\ldots ,R_t\}$ be a basis of $\spa\{R_1,\ldots ,R_q\}$.
If $q>t$, for $j>t$, we can write $R_j=\sum_{i=1}^t\alpha_{j,i}R_i$. In particular for $j=q$,
we can substitute this into Eqn.~\ref{eq:rexpression} and obtain

\begin{eqnarray*}
r(\mathbf{x}) & = &
\sum_i^{q-1}\left(
L_i(\mathbf{x})+\alpha_{q,i}L_{q}(\mathbf{x})
\right) R_i(\mathbf{x})
\end{eqnarray*}
where we let $\alpha_{q,i}=0$ for $i>t$. 
If $\{ L_1,\ldots ,L_q\}$ are independently, uniformly randomly
distributed, then
so are $\{ L_1+\alpha_{q,1}L_{q},\ldots , L_{q-1}+\alpha_{q,{q-1}}L_{q}\}$.
Continuing this process, we get that for
$\frac{n}{2}\geq t\geq \frac{n}{2}-\frac{3\sqrt{n}}{2}$,
there are sums $L_1',\ldots ,L_{t}',R_1',\ldots ,R_{t}'$
such that
\[
r(\mathbf{x})= \sum_{i=1}^{t}L_i'(\mathbf{x})R_{i}'(\mathbf{x})
\]
where the $\{L_1',\ldots ,L_{t}' \}$ are independently, uniformly random
and the $\{R_1',\ldots ,R_{t}' \}$ are linearly independent.
 We now further assume that 
\begin{equation}
\label{eq:assumptionhighrank2}
u=\mathrm{rk}(L_1',\ldots ,L_{t}')\geq t-\frac{3\sqrt{n}}{2}.
\end{equation}
Again, we will show at the end of this proof that this is true with
high probability. Applying a similar procedure as above, we get that
for some
\[
u
\geq t-\frac{3\sqrt{n}}{2}
\geq \frac{n}{2} -3\sqrt{n}
\]
there exist sums
$\widetilde{L}_1,\ldots , \widetilde{L}_{u}$ and
$\widetilde{R}_1,\ldots , \widetilde{R}_{u}$, such that
\[
r(\mathbf{x})=\sum_{i=1}^{u}\widetilde{L}_i(\mathbf{x})\widetilde{R}_i(\mathbf{x}),
\]
where all
$\widetilde{L}_1,\ldots ,\widetilde{L}_{u}$
and all
$\widetilde{R}_1,\ldots \widetilde{R}_{u}$ 
are linearly independent.
Thus, there exists a linear bijection
$(\mathbf{x}_1,\ldots ,\mathbf{x}_n)\mapsto
(\mathbf{z}_1,\ldots ,\mathbf{z}_n)$ with
$\mathbf{z}_1=\widetilde{L}_1,\ldots ,\mathbf{z}_{u}=\widetilde{L}_{u},
\mathbf{z}_{{u}+1}=\widetilde{R}_1,\ldots ,
\mathbf{z}_{2{u}}=\widetilde{R}_{u}$, such that 
\[
\widetilde{r}(\mathbf{z})
=
\mathbf{z}_1\mathbf{z}_{u+1}+\ldots ,\mathbf{z}_{{u}}\mathbf{z}_{2u}
\] 
where $r$ and $\widetilde{r}$ are equivalent up to a linear bijection on the inputs.
Since nonlinearity is invariant under linear bijections, we just need
to determine the nonlinearity of $\widetilde{r}$. 
Given that the inner product, $IP_n$, is a bent function,
it is elementary to verify that 

\[
NL(\widetilde{r})=
2^{n-2u}\left(2^{2u-1}-2^{u-1} \right)=
2^{n-1} - 2^{n-u-1}.
\]
If $u\geq \frac{n}{2}-3\sqrt{n}$, this is at least $2^{n-1} - 2^{\frac{n}{2}+3\sqrt{n}-1}$.

Now it remains to show that the probability of either
Assumption~(\ref{eq:assumptionhighrank1}) or
(\ref{eq:assumptionhighrank2}) occurring is so small that a union bound
over all the $2^n-1$
choices of $r$ gives that with high probability,
 \emph{every} $r\in \spa \{ \mathbf{y}_1,\ldots ,\mathbf{y}_n\}$
has at least the desired nonlinearity.

For Assumption~(\ref{eq:assumptionhighrank1}), we can think of the $q\geq \frac{n}{2}$
vectors $R_1,\ldots ,R_{q}$ as rows in a $q\times \frac{n}{2}$ matrix.
We will consider the  upper left $\frac{n}{2} \times \frac{n}{2}$
submatrix.
By Lemma~\ref{avishayrankprob}
this has rank at most $\frac{n}{2}-\frac{3\sqrt{n}}{2}$
with probability at most
\[
2^{\frac{n}{2}-\left(\frac{n}{2} - (\frac{n}{2}-\frac{3\sqrt{n}}{2}) \right)^2}
=
2^{\frac{n}{2} - \frac{9n}{4}}=
2^{-\frac{7n}{4}}.
\]
Similarly for Assumption~(\ref{eq:assumptionhighrank2}), we can consider the 
$\frac{n}{2}\geq t\geq \frac{n}{2}- \frac{3\sqrt{n}}{2}$
vectors $L_1',\ldots , L_{t}'$
as the rows in a $t\times \frac{n}{2}$ matrix.
Consider the top left $t\times t$ submatrix.
Again, by Lemma~\ref{avishayrankprob},
the probability of this matrix having rank at most
$t-\frac{3\sqrt{n}}{2}$ is at most
\[
2^{t-\left(t - (t-\frac{3\sqrt{n}}{2}) \right)^2}
\leq
2^{\frac{n}{2}-\frac{9n}{4}}
=
2^{-\frac{7n}{4}}.
\]
There are $2^n-1$ choices of $r$, so by the union bound,
the total probability of at least one of Assumption
(\ref{eq:assumptionhighrank1}) or (\ref{eq:assumptionhighrank2})
failing for a least one choice is at most
$2\cdot (2^n-1)\cdot 2^{-\frac{7n}{4}}$, which
tends to zero, so in fact the described construction will have
the desired nonlinearity with high probability.
\end{proof}
We should note that it is not hard to improve in the constants in the proof
and show that in fact the described function has nonlinearity at least
$2^{n-1}-2^{\frac{n}{2}+c\sqrt{n}}$ for some constant $c<3$. However,
the proof given does not allow improvement to e.g. $c=2$.

It follows from the Gilbert-Varshamov bound Theorem~\ref{thm:gilbertvar}) that $L(n,\frac{n-1}{2})<2.95n$ for large enough $n$. 

\begin{corollary}
For sufficiently large $n$ there exist $(n,n)$-functions with multiplicative complexity at most $2.95n$ with nonlinearity at least $2^{n-1}-2^{\frac{n}{2}+3\sqrt{n}}$.
\end{corollary}

\section{Multiplicative Complexity of an Explicit Vector Valued Function}
The multiplicative complexity of any
$(n,1)$-function
is between $0$ and $\left(1+o(1)\right)2^{n/2}$, as shown in \cite{nechiporuk1962}, (see also \cite{juknabook}), and a random function has multiplicative complexity at least $2^{n/2}-O(n)$~\citep{DBLP:journals/tcs/BoyarPP00} with probability $1-o(1)$.
However, there is no value of $n$ where a concrete $(n,1)$-function 
has been exhibited with a proof that more than $n-1$
AND gates are necessary to compute it.
A lower bound of $n-1$ follows by the simple \emph{degree bound}\footnote{Notice
that despite the name, this is not the same as
Strassen's degree bound as described in
\cite{strassen1973berechnungskomplexitat} and Chapter
8 of \cite{DBLP:books/daglib/0090316}.
}:
a function with degree $d$ has
multiplicative complexity at least $d-1$ \citep{Schnorr88}.

In this section, we first show that repeated use of the degree bound is sufficient to prove that an explicit function has multiplicative complexity at least $2n-3$. 

Furthermore, we show that any $(n,m)$-function has multiplicative complexity at most  $2.5(1+o_m(1))\sqrt{m2^n}$ and that this is tight up to a small constant factor.

\subsection{A Lower Bound for an Explicit Function}
Here we show that repeated use of the
degree bound gives a concrete $(n,n)$-function, exhibiting
a lower bound of $2n-3$. To the best of our knowledge this is
the best lower bound on the multiplicative complexity
for $(n,n)$-functions.
\begin{theorem}
\label{thm:2nlobo}
  The $(n,n)$-function $f$ defined as
$
f_i(\mathbf{x})= \prod_{j\in [n]\setminus \{ i\}} x_j
$, 
has multiplicative complexity at least $2n-3$.
\end{theorem}
\begin{proof}
Consider the first AND gate, $A$, with degree 
at least $n-1$. Such a gate exists since the outputs have degree $n-1$.
 By the degree bound, $A$ must have 
at least $p\geq n-3$ AND gates with degree at most $n-2$ in its
subcircuit. 
 Call these AND gates $A_1,\ldots ,A_p$.
None of these AND gates can be an output gate
 since they all compute functions of degree at most $n-2$
 and all outputs have degree $n-1$. Suppose
there are $q$ additional AND gates (including $A$), where some of these must have degree at least $n-1$.
Call these AND gates $B_1,\ldots , B_q$. Then, for every $i\in [n]$,
there exist $P_i\subseteq [p]$, $Q_i\subseteq [q]$, and $X_i\subseteq [n]$ such that
\[
f_i=\sum_{j\in P_i}A_j+\sum_{j\in Q_i}B_j+\sum_{j\in X_i}x_j.
\]
We can think of each $B_j$ (resp. $A_j$) as a vector in $\bbF_2^n$, where the $i$th coordinate is $1$ if the term
$\prod_{k\in [n]\setminus \{i \}}x_k$ is present
in the algebraic normal form of the function computed by $B_j$
(resp. $A_j$). Since
each $A_j$ has degree at most $n-2$, all the $A_j$ are zero vectors
in this representation, so
$\spa (A_1,\ldots ,A_p,B_1,\ldots ,B_q)=\spa (B_1,\ldots ,B_q)$.
It follows that
\[
\{f_1,\ldots,f_n \} \subseteq \spa (B_1,\ldots ,B_q).
\]
Therefore,
\[
n=\dim(\{f_1,\dots ,f_n \}) \leq \dim(\spa (B_1,\ldots ,B_q))\leq q.
\]
We conclude that the circuit has at least $q+p\geq 2n-3$ AND gates.
\end{proof}

The multiplicative complexity of the function is at most $3n-6$. This can be seen from the following construction:

\begin{enumerate}
\item Use $n-3$ AND gates with the following outputs:
\[ A'=\{ x_1x_2, x_1x_2x_3,...,x_1x_2\ldots x_{n-2}\}.\]
From these, produce the output, $x_1x_2\ldots x_{n-1}$, 
called $A$, in the previous proof, with one additional AND gate.
Note that no other AND gates are used in the subcircuit computing $A$,
so the following gates are among $(B_1,\ldots ,B_q)$ from that proof.
\item Use $n-2$ AND gates with the following outputs:
\[ B'=\{ x_2x_3\ldots x_n,x_3x_4\ldots x_n,...,x_{n-1}x_n\}.\]
\item Use $n-4$ AND gates to AND together the $i$th element
of $A'$ with the $i+2$nd element of $B'$, for $1\leq i\leq i-4$.
\item Compute $x_1 \cdot (x_3x_4\ldots x_n)$ and
$(x_1x_2\ldots x_{n-2}) \cdot x_n$.
\end{enumerate}
We leave it as an interesting open question to close this gap.

\subsection{Multiplicative Complexity of $(n,m)$-functions: Upper and Lower bounds}
Below we look at the multiplicative complexity of the hardest
$(n,m)$-functions. We give a construction showing that any such function has multiplicative complexity at most 
$2.5(1+o_m(1))\sqrt{m2^n}$, where $o_m(1)$ denotes a function that tends to $0$ when $m$ goes to infinity. For some values of $n,m$ the construction gives the slightly better bound $2(1+o_m(1))\sqrt{m2^n}$. A counting argument shows that this is at most a small factor from being tight.

\begin{theorem}
 Let $f$ be a random $(n,m)$-function, $m\leq 2^n$. Then, almost every $f$ satisfies,
\[
 c_\wedge(f)\geq \sqrt{m 2^n}-2n- \frac{m}{2}.
\]
\end{theorem}
This proof is similar to the proof of Lemma~15 in \cite{DBLP:journals/tcs/BoyarPP00}.
\begin{proof}
First we give an upper bound on the number of functions that can be computed with circuits using at most $k$ AND gates. Let $A_1,A_2,\ldots ,A_k$ be some topological ordering of the AND gates.
The two inputs to $A_i$ are the XORs of some of the previous AND gates, and some of the variables in the circuit. Without loss of generality, we assume that the constant $1$ is only used at the output gates of the circuit.
Thus, the number of choices for $A_i$ is
\[
\left( 2^{i-1+n} \right)^2/2=
 2^{2i+2n-3}.
\]
Each output gate of the circuit is the XOR of some of the AND gates, some of the variables to the circuit, and possibly the constant $1$. That is, for $k\geq 0$ the total number of ways to choose the inputs to the AND gates and the outputs is:
\[
\left(2^{k+n+1}\right)^m\prod_{i=1}^k  2^{2i+2n-3}
=
2^{k^2+m \cdot (1+n)+k\cdot(m+2n-2)}.
\]
Since we assume that $m\leq 2^n$, we have that for sufficiently large $n$,
 $\sqrt{m 2^n}-2n- \frac{m}{2}>0$ ($n\geq 12$ suffices). So the number of
$(n,m)$-functions with multiplicative complexity at most
$\sqrt{m 2^n}-2n- \frac{m}{2}$ is at most 
\begin{align*}
2^{m2^n +2m-m^2/4+4n-\sqrt{m}2^{n/2+1}(1+n)}.
\end{align*}
There are $2^{m2^n}$ different $(n,m)$-functions, so the probability that a random
$(n,m)$-function can be computed with a circuit with at most 
$\sqrt{m 2^n}-2n- \frac{m}{2}$ AND gates is at most
\[
2^{2m-m^2/4+4n-\sqrt{m}2^{n/2+1}(1+n)},
\]
which tends to $0$ when $n$ goes to infinity.
\end{proof}

On the other hand we present an almost matching upper bound.
The technique has similarities to those used in \citep{lupanov1958method}.

\begin{theorem}
\label{thm:nmfunctionsmcupperbound}
Let $f$ be an $(n,m)$-function. If $\log m$ is an integer and $n+\log m$ is even, then
\[
c_\wedge(f) \leq 
2 (1+o_m(1))\sqrt{m2^n},
\]
otherwise 
\[
c_\wedge(f) \leq 
 2.5(1+o_m(1))\sqrt{m2^n}.
\]
\end{theorem}
Before presenting the proof, we define indicator functions, and a result about their
multiplicative complexity.

\begin{definition}
 For every $n\in \mathbb{N}$ and $\mathbf{z}\in \bbF_2^n$ the indicator function
$I_\mathbf{z}\mapnbit$ is defined as $I_\mathbf{z}(\mathbf{x})=1$ if and only if
$\mathbf{z=x}$.
\end{definition}
The following simple proposition on the multiplicative complexity of indicator functions
will be helpful in the proof.

\begin{proposition}
Let $n>1$ be arbitrary. Define the $(n,2^n)$-function
\[
 AI_n(\mathbf{x})
=
(
I_{(0,0,\dots , 0)}
(\mathbf{x}),
I_{(1,0,\dots , 0)}
(\mathbf{x}), 
\ldots,
I_{(1,1,\dots , 1)}
(\mathbf{x})
).
\]
Then the multiplicative complexity of $AI_n$ is 
$c_\wedge(AI_n)=2^{n}-n-1$
\end{proposition}
\begin{proof}
 First, we show that $c_\wedge(AI_n)\leq 2^{n}-n-1$.
 We start by computing all quadratic terms, that is terms on the form
 ${x}_i{x}_j$ for $1\leq i<j\leq n$. This can be done using one
 AND gate for each of the $\binom{n}{2}$ terms. Now we compute each degree three term.
 Since each degree three term $x_ix_jx_k$ can be written as $x_iQ$ for some quadratic
 term $Q$, we can do this with one AND gate for each term. We continue in this way until we have computed the term $\prod_{i\in S}x_i$ for each $S\subseteq [n]$.
 The number of AND gates used is 
 \[
  \binom{n}{2} + \binom{n}{3} + \ldots + \binom{n}{n-1}+\binom{n}{n}=2^{n}-n-1.
 \]
 Now that all the terms that can occur in an ANF have been obtained, \emph{any}
 function can be computed without using additional AND gates.
In particular, all the indicator functions can be computed using no additional AND gates.

For the lower bound, suppose $M$ AND gates, $A_1,\ldots ,A_M$, suffice to compute the
$2^n$ indicator functions. Let $T_1,\ldots ,T_{2^n}$ be some ordering
of the terms $\prod_{i\in S}x_i$, $S\subseteq [n]$. 
Each function $f\mapnbit$ can be considered
as a vector in $\bbF_2^{2^n}$ by letting the $i$th coordinate be $1$ if and only if
term $T_i$ is included in the ANF of $f$. Considering all functions as vectors in this
way, it follows that
\[
 \{ I_{\mathbf{z}} | \mathbf{z}\in \bbF_2^n \} \subseteq 
 \spa_{\bbF_2} \{1,x_1,\ldots ,x_{n} , A_1,\ldots , A_M \}.
\]
All the indicator functions are linearly independent, so we have
\[
 2^n = \dim( \{ I_{\mathbf{z}} | \mathbf{z}\in \bbF_2^n \})\leq 
\dim(\spa_{\bbF_2} \{1,x_1,\ldots ,x_{n} , A_1,\ldots , A_M \})
\leq M+n+1.
\]
\end{proof}

\begin{proof}[Proof of Theorem~\ref{thm:nmfunctionsmcupperbound}.]
In the following let $k$ be an integer to be determined later. First compute all the indicator functions on the last $n-k$ variables
$x_{k+1},x_{k+2},\ldots , x_n$. This uses $2^{n-k}-(n-k)-1$ AND gates. Given all these indicators, using only XOR gates, it is possible to compute the function
\[
  f_i(a_1,a_2,\ldots ,a_k,x_{k+1},x_{k+2},\ldots ,x_n),
\]
for any choice of constants $i\in [m]$ and $\mathbf{a}\in \bbF_2^{k}$.
Now compute all the indicator functions on the first $k$ variables using
$2^k-k-1$ AND gates. After this, for each $i\in [m]$, and for each $\mathbf{a}\in \bbF_2^{k}$ compute
\[
g_{i,\mathbf{a}}(\mathbf{x}):=I_{(a_1,\ldots ,a_k)}(x_1,\ldots ,x_k)\cdot f_i(a_1,a_2,\ldots ,a_k,x_{k+1},x_{k+2},\ldots ,x_n).
\]
This uses $m\cdot 2^{k}$ AND gates. Now observe that
\[
 f_i(\mathbf{x})
 =\sum_{\mathbf{a}\in \bbF_2^{k}} g_{i,\mathbf{a}}(\mathbf{x}),
\]
so each $f_i$ can be obtained using only XOR gates. The total number of AND gates used is less than
\[
 m2^k + 2^{n-k} + 2^k.
\]
Suppose that $\log m$ is an integer and $n+\log m$ is even. Letting
$k=\frac{n-\log m}{2}$ results in
\[
 m2^{\frac{n-\log m}{2}} + 2^{\frac{n+\log m}{2}} + 2^{\frac{n-\log m}{2}} = 
 (1+o_m(1))2\sqrt{m2^n}.
\]
Otherwise we let $k=\left\lceil \frac{n-\log m}{2} \right\rceil$.
Let $k=\frac{n-\log m}{2} + \epsilon$. Then the total number of AND gates is at most
\[
 (1+o_m(1))
 \sqrt{m} 2^{ \frac{n}{2}}
 (
  2^{\epsilon}
  +
  2^{- \epsilon }
  )
 \leq
 2.5(1+o_m(1))\sqrt{m}2^{\frac{n}{2}}.
\]
\end{proof}

\section{Open Problems}
Strassen~\citep{strassen1973vermeidung} (see also \cite{DBLP:books/daglib/0090316}, Proposition 14.1, 
p. 351) proved that for an \emph{infinite field}, $\mathbb{K}$, if the quadratic function
$F\colon \mathbb{K}^n\rightarrow \mathbb{K}^m$ 
can be computed with $M$
multiplications/divisions, then it can be computed in $M$ multiplications by a
quadratic circuit.
However, it is unknown whether a
similar result holds for \emph{finite fields} and in particular for $\bbF_2$.
Mirwald and Schnorr \citep{DBLP:journals/tcs/MirwaldS92}
showed that for quadratic $(n,1)$- and $(n,2)$-functions, 
quadratic circuits are optimal. It is still not known
whether this is true for $(n,m)$-functions in general.
 It would be very interesting
to determine if the bounds proven here for quadratic circuits also
hold for general circuits.

When inspecting the proof of Theorem~\ref{thm:nlecc}, 
one can make a weaker assumption on the circuit than it having $\Sigma\Pi\Sigma$ structure.
For example, it is sufficient if it holds that for every AND gate, $A$, there
is a unique AND gate, $A'$ (which might be equal to $A$),
such that every path from $A$ to an output
goes through $A'$.
Can one find a larger, interesting class of circuits where
the proof holds?

The function defined in Theorem~\ref{thm:2nlobo} has multiplicative complexity at
least $2n-3$ and at most $3n-6$. 
What is the exact value?

\section{Acknowledgements}
A preliminary version of this paper appeared in the
proceedings of the 39th International Symposium on. Mathematical Foundations of Computer Science, 2014.
Supported in part by the Danish Council for Independent Research, Natural
Sciences, grant DFF-1323-00247.

\section*{References}

\bibliographystyle{elsarticle-harv} 

\bibliography{journals,nlmc}

 \end{document}